\documentclass[11pt,a4paper]{article}

\usepackage{amsmath,amsthm,amssymb,amsfonts}

\usepackage{authblk}
\usepackage{hyperref}
\usepackage{setspace}
\usepackage{mathrsfs}
\usepackage{cases}
\usepackage{graphicx}
\usepackage{enumerate}
\usepackage{subcaption} 
\usepackage{datetime}

\numberwithin{equation}{section}
\newcommand{\ii}{{\rm i}}
\newcommand{\ee}{{\rm e}}

\newcommand{\x}{{\rm x}}

\newtheorem{thm}{Theorem}
\newtheorem{lemma}[thm]{Lemma}

\theoremstyle{definition}

\usepackage{color}

\newdateformat{daymonthyear}{\THEDAY \, \monthname[\THEMONTH] \THEYEAR}
\newdateformat{monthyear}{\monthname[\THEMONTH] \THEYEAR}

\begin{document}

\title{Semi-classical gravity in de Sitter spacetime and the cosmological constant}
\author{Benito A. Ju\'arez-Aubry\thanks{{\tt benito.juarez@iimas.unam.mx}}}
\affil{Departamento de F\'isica Matem\'atica \\
Instituto de Investigaciones en Matem\'aticas Aplicadas y en Sistemas, Universidad Nacional Aut\'onoma de M\'exico,\\A. Postal 70-543, Mexico City 045010, Mexico}
\date{\daymonthyear\today}

\maketitle

\centerline{Published in Phys.~Lett.~B, {\bf 797} (2019) 134912 \href{https://doi.org/10.1016/j.physletb.2019.134912}{doi:10.1016/j.physletb.2019.134912}
}

\begin{abstract}
We show that there exist solutions to the semi-classical gravity equations in de Sitter spacetime sourced by the renormalised stress-energy tensor of a free Klein-Gordon field. For the massless scalar, solutions exist for every possible value of the cosmological constant, provided that the curvature coupling parameter is chosen appropriately. In the massive case, imposing Wald's axioms for the renormalised stress-energy tensor, the mass of the field and the curvature coupling constraint the allowed values of $\Lambda$. For a massive, minimally coupled field, a ``small $\Lambda$" solution is found, fixed by the relation $m^2 \simeq 4.89707 \times 10^{12} \Lambda$. We emphasise that in this framework, the {\it old} cosmological constant problem in its standard formulation plays no r\^ole, in the sense that there are no bare-vs-physical values of $\Lambda$, for only the physical $\Lambda$ appears in the semi-classical equations, and the value that it is allowed to take is fixed or restricted by the equations themselves. We explain that there is an important relation between the cosmological constant problem and the violation of Wald's stress-energy renormalisation axioms.
\end{abstract}

\singlespacing

\section{Introduction}
\label{sec:Intro}

The standard model of cosmology, $\Lambda$CDM, is currently the most successful theoretical model for explaining the evolution of the universe as a whole from the CMB epoch to the current era. While there is no doubt about its phenomenological success, it has left room for  several puzzles in theoretical physics, including the so-called cosmological constant problem. 

The old cosmological constant problem is in its standard formulation a naturalness one, which posits that some bare cosmological constant, $\Lambda_{\rm bare}$, of the order of the very large energy density of the quantum fields of matter, should be renormalised with exquisite precision to the small, positive value that we observe, $\Lambda_{\rm ren}$. See e.g. the excellent references \cite{Weinberg:1988cp, Padilla:2015aaa, Martin:2012bt}.



In this letter, we study the semi-classical gravity equations in de Sitter spacetime for a free Klein-Gordon field. On the one hand, this is an interesting problem in its own right in the sense that few exact semi-classical gravity solutions are known. The author is aware of \cite{Dappiaggi:2008mm}, in which cosmology with a massive Klein-Gordon field with $\xi = 1/6$ exhibits a de Sitter phase. Due to the isometries of de Sitter spacetime, we are able to show the existence of a large number of solutions parametrised by the mass parameter of the Klein-Gordon field, $m^2$, and the curvature coupling parameter, $\xi$. On the other hand, the cosmological constant puzzle can be most naturally framed in the setting of semi-classical gravity, without necessarily making reference to Minkowskian mode expansions or energy cut-offs as in the classical statement on the problem \cite{Weinberg:1988cp, Padilla:2015aaa, Martin:2012bt}, while also offering interesting insights on the puzzle, including its connection with the violation of Wald's stress-energy renormalisation axioms.

In the semi-classical setting, as we shall see below, bare and renormalised quantities for $\Lambda$ simply do not appear in the same way that the bare stress-energy tensor plays no r\^ole in the equations. Instead, $\Lambda$ should take values that are consistent with solutions to the semi-classical Einstein field equations. The point that we shall make below, in the context of a simple Klein-Gordon model, is that in some cases the values of $\Lambda$ is highly restricted, and sometimes uniquely determined (up to renormalisation ambiguities), in terms of the parameters of the Klein-Gordon field, $m^2$ and $\xi$. Moreover, as we shall see, while the value of $\Lambda$ will be necessarily proportional to $m^2$ (in the massive case), the proportionality factor can account for several orders of magnitude in difference between the two parameters. In particular, in the case of a massive, minimally coupled field, one finds an interesting solution that fixes the ratio $m^2/\Lambda \sim 10^{12}$. This framework offers an interesting perspective to address the cosmological constant puzzle, by studying the semi-classical Einstein equations in a realistic model of our universe -- FLRW spacetime with the Standard Model of Particle Physics as matter content.

On this tone, we should clarify that an aim of this paper is not to obtain the value for the cosmological constant observed in our universe, but we do argue that in the framework of the semi-classical Einstein field equations, and imposing Wald's stress-energy renormalisation axioms, its peculiar, small, observational value is not at odds with any estimated would-be, large, unrenormalised value.

This letter is organised as follows: After this introduction, in Sec. \ref{sec:semiGR} we briefly review the elements of semi-classical gravity. In Sec. \ref{sec:deSitter} we specify to the problem of a Klein-Gordon field in de Sitter spacetime. An analysis of the ambiguities is made, where we show that some of them can be fixed in de Sitter, both algebraically and by imposing Wald's axioms.  In Sec. \ref{sec:Solutions}, exploiting the symmetries of de Sitter, we show that there exists solutions to the problem stated in Sec. \ref{sec:deSitter}, and that for such solutions the cosmological constant will (almost always and modulo ambiguities) be restricted to specific values in terms of the $m^2$ and $\xi$ parameters of the Klein-Gordon field. Our final remarks are made in Sec. \ref{sec:Remarks}, including comments on the old cosmological constant problem, and its relation to Wald's renormalisation axioms of the stress-energy tensor \cite{Wald:1995yp}.

\section{Semi-classical gravity preliminaries}
\label{sec:semiGR}

We wish to obtain solutions to the semi-classical Einstein field equations for a quantum Klein-Gordon field in de Sitter spacetime. The semi-classical gravity equations are
\begin{subequations}
\label{SemiGR}
\begin{align}
	& G_{ab} + \Lambda g_{ab} = 8\pi G_{\rm N} \langle \Psi | \hat T_{ab} \Psi \rangle, \label{SemiEinstein} \\
	& (\Box - m^2 - \xi R) \hat \Phi = 0, \label{SemiFields} 
\end{align}
\end{subequations}
where the classical geometry is sourced by the renormalised quantum stress-energy tensor of the matter field. Here, $\Lambda > 0$ is the cosmological constant. 

The quantum Klein-Gordon field is an operator-valued distribution, $f \mapsto \hat{\Phi}(f)$, for $f \in C_0^\infty(M)$, which is densely defined on the relevant Fock space of the theory, $\hat \Phi(f): \mathscr{D} \subset \mathscr{H} \to \mathscr{H}$. 
The set of field observables generates a unital operator $^*$-algebra, $\mathscr{A}_{\rm KG}$, generated by smeared fields and subject to the following relations: for $f, g \in C_0^\infty(M)$, (i) $f \mapsto \hat \Phi(f)$ is linear, (ii) $\hat \Phi(f)^* = \hat \Phi(f)$ the field is self-adjoint, (iii) $\hat \Phi ((\Box - m^2 - \xi R)f) = 0$, the field equation \eqref{SemiFields} holds by integration by parts twice and (iv) $[\hat \Phi(f), \hat{\Phi}(g)] = -\ii E(f,g) 1\!\!1$ the field satisfies commutation relations, where $E = E^- - E^+$ is the advanced-minus-retarded fundamental Green operator of the Klein-Gordon operator $\Box - m^2 - \xi R$. 

The discussion on how to compute the renormalised stress-energy tensor that appears on the right-hand side of \eqref{SemiEinstein} is a well-trodden path for the Klein-Gordon field. Our purpose is therefore to give a short, non-exhaustive review and we refer the reader to the classical literature \cite{Wald:1995yp} for the details.\footnote{The interested reader might also look at \cite{Moretti:2001qh, Christensen:1976vb, Wald:1978pj, Decanini:2005eg}.} The classification of renormalisation ambiguities (subject to some regularity criteria), which plays a key r\^ole for semi-classical gravity, was laid out in \cite{Hollands:2001nf, Hollands:2001fb, Khavkine:2014zsa, Hollands:2004yh}. 

The starting point is computing the two-point function in a Hadamard state, $\Psi$, which in a geodesically convex subset, $O \subset M$, in which $\sigma(\x,\x')$, the half-squared geodesic distance is well defined, takes the Hadamard form
\begin{equation}
G^+(\x,\x') = \frac{1}{8\pi^2}\left[\frac{\Delta^{1/2}(\x,\x')}{\sigma_\epsilon(\x,\x')} + v(\x,\x') \ln \left(\sigma_\epsilon(\x,\x')/\ell^2 \right) + w(\x, \x')\right].
\label{G+Had}
\end{equation}

Here, $\sigma_\epsilon(\x,\x') = \sigma(\x, \x') + 2 \ii T(\x, \x') \epsilon + \epsilon^2$ is the regularised half-squared geodesic distance, $\Delta$ is the van Vleck-Morette determinant and $v$ and $w$ are smooth bi-functions computed as a covariant Taylor series in powers of $\sigma$ through the Hadamard recursion relations, obtained by demanding that $G^+(\x,\x')$ be a solution of the Klein-Gordon equation in the $\x$-variable, provided that an intial value $w_0$ for the $O(1)$ term in the $w$-series is prescribed. The initial values of $u$ and $v$ are determined geometrically, and independent of the quantum state. The datum $w_0$ is state dependent.

The renormalised stress-energy tensor is obtained by acting with a differential two-point operator on the singularity-subtracted two-point function, the smooth bi-function defined by $G^+_{\rm reg}(\x, \x') = G^+(\x,\x') - H_\ell(\x,\x')$, where
\begin{equation}
H_\ell(\x,\x') = \frac{1}{8\pi^2}\left[\frac{\Delta^{1/2}(\x,\x')}{\sigma_\epsilon(\x,\x')} + v(\x,\x') \ln \left(\sigma_\epsilon(\x,\x')/\ell^2 \right) + w^{\rm Had}(\x, \x')\right]. \label{HadamardF}
\end{equation}
is the Hadamard parametrix, with $\ell$ a fixed length scale and $w^{\rm Had}$ as the $w$ smooth bi-function obtained from the initial value $w_0 = 0$. We define the renormalised stress-energy tensor by
\begin{subequations}
\label{TRenWald}
\begin{align}
\langle \Psi | \hat T_{ab} (\x) |\Psi \rangle   & = \lim\limits_{\x' \rightarrow \x} \mathcal{T}_{ab} \left[G^+_{\rm reg} - \frac{1}{8 \pi} g_{ab} v_1(\x, \x') \right]  + \Theta_{ab} (\x), \\ 
\mathcal{T}_{ab} &= (1-2\xi ) g_{a}\,^{b'}(\nabla_a) (\nabla_{b'} )+\left(2\xi - \frac{1}{2}\right) g_{ab}g^{cd'} (\nabla_c) (\nabla_{d'})  - \frac{1}{2} g_{ab} m^2 \nonumber\\
&\quad + 2\xi \Big[  - g_{a}\,^{a'} g_{b}\,^{b'} \nabla_{a'} \nabla_{b'} + g_{ab} g^{c d}\nabla_c \nabla_d + \frac{1}{2}G_{ab} \Big], 
\end{align}
\end{subequations}
where derivatives with primed and unprimed indeces are evaluated at the points $\x'$ and $\x$ respectively. In eq. \eqref{TRenWald}, $g_{ab'}$ denotes the bi-vector of parallel transport from $\x$ to $\x'$, with the condition $\lim_{\x' \rightarrow \x}g_{ab'}=g_{ab}$ \cite{DeWitt:1960fc}, the term $[v_1]_c(\x) = \lim_{\x' \to \x} v_1(\x, \x')$, i.e. the diagonal of the $v_1$ coefficient in the Hadamard series of $v$, is given by \cite{Wald:1978pj, Decanini:2005eg} $[v_1]_c  = \lim\limits_{\x' \rightarrow \x} v_1(\x, \x') = \frac{1}{8} m^4 + \frac{1}{4} \left(\xi - \frac{1}{6}\right) m^2 R - \frac{1}{24}\left(\xi - \frac{1}{5} \right) \Box R 
 \quad + \frac{1}{8} \left(\xi - \frac{1}{6} \right)^2 R^2 - \frac{1}{720} R_{ab} R^{ab} + \frac{1}{720} R_{abcd} R^{abcd}$, and $\Theta$ is an ambiguous, geometric, covariantly conserved, symmetric tensor of dimension of lenght to the minus fourth power, built out of the metric and its derivatives, which has been classified in \cite{Hollands:2001nf, Hollands:2001fb, Khavkine:2014zsa}. For conformally-coupled fields, $[v_1]_c$ is responsable for the trace anomaly \cite{Wald:1978pj}.

As a final word for the section, notice that $[v_1]_c$ spoils the second order, hyperbolic form of the semi-classical system \eqref{SemiGR}. This is a well-known problem in semi-classical gravity. However, as we shall see below, in the symmetry-reduced case of de Sitter spacetime, this problem does not occur.

\section{Semi-classical gravity in de Sitter spacetime}
\label{sec:deSitter}

The metric tensor for the $(3+1)$-dimensional de Sitter spacetime has the form $g = (\alpha/\eta)^2 (- d\eta^2 + dx^2+dy^2+dz^2)$, with $\alpha^2 = 3/\Lambda$. Eq. \eqref{SemiFields} has been studied in de Sitter spacetime in \cite{Bunch:1978yq} by exploiting the spatial symmetries of the problem, whereby the wave equation reduces to an ODE for the temporal part that can be brought to a Bessel equation form. The quantum fields can be concretely represented as operators in the Hilbert space, $\mathscr{H}_{\rm BD}$ as,
\begin{subequations}
\begin{align}
\hat \Phi(\eta, x) & = (2 \pi)^{-3/2} \int d^3k \left( \psi_k \ee^{\ii k \cdot x} \hat a_k + \overline{\psi_k} \ee^{-\ii k \cdot x} \hat a^*_k \right), \\
\psi_k(\eta) & = \alpha^{-1} (18 \pi)^{1/2} \eta^{3/2} H_\nu^{(2)}(k \eta), \hspace{0.5cm} \nu^2 = 9/4 - 12(m^2/R + \xi).
\end{align}
\end{subequations}

Annihilation operators annihilate $\Omega_{\rm BD} \in \mathscr{H}$, the Bunch-Davies vacuum, which is cyclic in the sense that the set of vectors obtained from acting with all operator algebra elements contained in $\mathscr{A}_{\rm KG}$ on the vacuum $\Omega_{\rm BD}$ forms a dense subspace of the Fock space $\mathscr{H}$. Creation and annihilation operators act in the usual way on Fock space elements, see e.g. \cite[App. A.3]{Wald:1995yp}.


The two-point function, which characterises the Bunch-Davies vacuum (as it is quasi-free), can then be obtained directly as a sum over modes, $G^+(\x, \x')  = (2 \pi)^{-3} \int_{\mathbb{R}^3} \! d^3 k \, \psi_k(\eta) \overline{\psi_k(\eta')} \ee^{\ii k \cdot (x - x')}$, and it admits a closed form expression in terms of hypergeometric functions
\begin{equation}
G^+(\x, \x')= \frac{\sec \left[\pi \nu (1/4 - \nu^2) \right]}{16 \pi \alpha^2} F\left[\frac{3}{2}+\nu; 2; 1 + \frac{(\eta-\eta')^2 - |r-r'|^2}{4 \eta \eta'} \right].
\label{BDGreen}
\end{equation}

We should mention that, in the algebraic approach to quantum field theory, the (algebraic) state is defined by all its $n$-point functions, as well as a normalisation and a positivity requirements. The concrete operators on a Hilbert space representations are obtained via the GNS construction. For vacuum states, all $n$-point functions are reconstructed from the two-point function -- they are quasi-free. Hence, eq. \eqref{BDGreen} can be taken as the definition of the Bunch-Davies vacuum and the starting point of quantum field theory in de Sitter spacetime, together with the {\it abstract} Klein-Gordon algebra.

In \cite{Bunch:1978yq} the point-splitting and renormalisation of the stress-energy tensor in the Bunch-Davies vacuum has been performed with the aid of DeWitt-Schwinger, rather than Hadamard expansions (cf. Sec. \ref{sec:semiGR}). The two renormalisation methods are well known to be equivalent. See in particular \cite{Decanini:2005gt} for a detailed discussion. The results reported in \cite{Bunch:1978yq} are tantamount to performing the Hadamard renormalisation at the fixed scale $\ell^2 = m^{-2}$, yielding
\begin{align}
& \langle \Omega_{\rm BD} |  \hat T_{ab} \Omega_{\rm BD} \rangle  = \frac{g_{ab}}{(8 \pi)^2} \left\{m^2\left[m^2 + \left( \xi - 1/6\right) R \right]\left[\psi\left(3/2 + \nu \right) + \psi\left(3/2 - \nu \right) \right. \right. \nonumber \\
& \left. \left. - \! \ln\left(12 m^2/R \right) \right] \! - \frac{1}{2}(\xi - 1/6)^2 R^2 \! + \frac{R^2}{2160} \! - m^2(\xi -1/6)R \! - \frac{m^2 R}{18}\right\} \! + \Theta_{ab}, \label{TabBD}
\end{align}
where $\psi$ is the digamma function, defined by $z \mapsto \psi(z) = \Gamma'(z)/\Gamma(z)$, where $\Gamma$ is the Gamma function. In ref. \cite{Bunch:1978yq} the ambiguity term, $\Theta_{ab}$, is ignored, but we have restored it, as it will play an important r\^ole in the ensuing discussion.

\subsection{Stress-energy tensor ambiguities}

The tensor $\Theta_{ab}$ has ambiguity terms coming from two sources. First, changing the natural renormalisation scale $\ell^2 = m^{-2}$ to some arbitrary scale $\ell^2 = \mu^{-2}$, one has a scale ambiguity contribution, $\Theta_{ab}^{m^2/\mu^2}$. Second, from the classification of allowed ambiguous terms discussed in \cite{Hollands:2001nf, Hollands:2001fb, Khavkine:2014zsa}, one has an additional $\Theta_{ab}^{\rm clas}$. The total contribution to the ambiguity tensor is therefore $\Theta_{ab} = \Theta_{ab}^{m^2/\mu^2} + \Theta_{ab}^{\rm clas}$.

Let us first discuss the ambiguity term $\Theta_{ab}^{m^2/\mu^2}$. As we have said, for the massive scalar field, in obtaining eq. \eqref{TabBD}, the natural choice of Hadamard parametrix has been made with $\ell^2 = m^{-2}$. Notice, that in the massless case the Hadamard parametrix contains no ambiguities, and hence the stress-energy tensor vanishes in Minkowski spacetime identically. 

We henceforth treat the renormalisation scale as an arbitrary one for the massive field, $\ell^2 = \mu^{-2}$, and subtract the Hadmard parametrix $H_{\mu^{-1}}$, cf. eq. \eqref{HadamardF}, yielding an additional term  on the right-hand side of eq. \eqref{TabBD}, cf. \cite[eq. (110)]{Decanini:2005eg}, given by
\begin{align}
\Theta_{ab}^{m^2/\mu^2} & = \frac{\ln(m^2/\mu^2)}{2(2\pi)^2}g_{ab} \left(\frac{1}{8}m^4 +\frac{1}{2}\left(\xi - \frac{1}{6} \right) m^2 \Lambda  \right).
\label{AmbiguityLambda}
\end{align}

We now dicuss the contribution $\Theta_{ab}^{\rm clas}$. In agreement with local covariance, stress-energy conservation, correct scaling properties, as well as other regularity criteria, the ambiguity term $\Theta_{ab}^{\rm clas}$ has the form $\Theta_{ab}^{\rm clas}  \! = \! \alpha_1 \left(- \frac{1}{2} g_{ab} R^2 + 2 R R_{ab}\right) +  \alpha_2 \left(2 R_{a}{}^c R_{cb} - \frac{1}{2} g_{ab} R^{cd} R_{cd} \right) + \alpha_3 m^2 G_{ab} + \alpha_4 m^4 g_{ab}$, where the $\alpha_i: \xi \mapsto \alpha_i(\xi)$, $i =\{1,2,3,4\}$, are renormalisation ambiguities that will depend on $\xi$. Trivial dependence means that the $\alpha_i$ are constants. That $\Theta_{ab}^{\rm clas}$ must be of this form is implied by the results in \cite{Hollands:2004yh}, see in particular Remark (3) towards the end of Sec. 5.1 of that reference, and the general form\footnote{Since we are interested in de Sitter spacetime, we have ommitted the terms that appear as derivatives of the curvature.} of this ambiguity had been suspected to be as shown in \cite{Hollands:2004yh} at least since the late seventies \cite{Wald:1978pj}.

In constant curvature spacetimes, the terms that accompany the coefficients $\alpha_1$ and $\alpha_2$ can be seen to vanish algebraically because $R_{abcd} = (\Lambda/3)^2 (g_{ac}g_{bd}-g_{ad}g_{cb})$. One is therefore left with
\begin{align}
\Theta_{ab} & = g_{ab} \left[ \left( \frac{\ln(m^2/\mu^2)}{(8\pi)^2} + \alpha_4 \right) m^4 + \left( \frac{\ln(m^2/\mu^2)}{(4\pi)^2}\left(\xi - \frac{1}{6} \right) - \alpha_3\right) m^2 \Lambda \right], \label{RenAmb}
\end{align}

\subsection{Wald's stress-energy renormalisation axioms}

We take the viewpoint that the stress-energy tensor for the Klein-Gordon field in de Sitter should satisfy Wald's axioms \cite{Wald:1995yp}. This will help fix certain ambiguities coefficients in $\Theta_{ab}$. The axioms are: (i) If $\Psi_1$ and $\Psi_2$ are Hadamard state vectors, then $\langle \Psi_1 |  \hat T_{ab} \Psi_1 \rangle - \langle \Psi_2 |  \hat T_{ab} \Psi_2 \rangle \in C^\infty(M)$. (ii) Local covariance. (iii) If $\Psi$ is a Hadamard state vector, $\nabla^a \langle \Psi |  \hat T_{ab} \Psi \rangle = 0$. (iv) In Minkowski spacetime $\langle \Omega_{\rm M} |  \hat T_{ab} \Omega_{\rm M} \rangle = 0$.

Axioms (i) - (iii) are satisfied by construction. Axiom (iv) can be imposed in the following way. In the $\Lambda \to 0^+$ limit, the renormalised stress-energy tensor reduces to 
\begin{equation}
0 = \langle \Omega_{\rm M} | \hat T_{ab} \Omega_{\rm M} \rangle = \lim_{\Lambda \to 0^+} \langle \Omega_{\rm BD} | \hat T_{ab} \Omega_{\rm BD} \rangle = \left( \frac{\ln(m^2/\mu^2)}{(8\pi)^2} + \alpha_4 \right) m^4 g_{ab}.
\label{MinkoAmbLimit}
\end{equation}

Hence, $\Theta_{ab}$ has the explicit form $\Theta_{ab} = - \alpha_\mu(\xi) m^2 \Lambda g_{ab}/(8\pi)^2$, where we have set $\alpha_\mu(\xi) = (8 \pi)^2 \left( \alpha_3(\xi) - \ln(m^2/\mu^2)\left(\xi - \frac{1}{6} \right)/(4\pi)^2 \right) $.

\section{Existence of solutions in de Sitter spacetime}
\label{sec:Solutions}

We now seek solutions to eq. \eqref{SemiGR}. Due to the large symmetry of the problem, the task is reduced to solving an algebraic relation. In turn, this relation will provide the admissible values for $\Lambda$ in terms of the parameters of the Klein-Gordon field theory, $m^2$ and $\xi$. 

\subsection{The massless case}

In the massless case, $m^2 = 0$, there are solutions for any $\Lambda > 0$ provided that $\xi$ takes the values $\xi_+ = 1/6 + (1080)^{-1/2}$ or $\xi_- = 1/6 - (1080)^{-1/2}$.

\subsection{The massive case}

For the massive case, set $x = m^2/(4 \Lambda)$. The solutions lie at the roots of the function $f_\xi: \mathbb{R}^+ \to \mathbb{R}$ defined by
\begin{align}
f_\xi(x) & = \psi\left(3/2 + \nu(x) \right) + \psi\left(3/2 - \nu(x) \right) - \ln\left(12 x \right) \nonumber \\
&  + \left[x + \left( \xi - 1/6\right) \right]^{-1}\left[ \frac{1-1080(\xi - 1/6)^2}{2160x} - \xi + \frac{1}{9}- \frac{\alpha_\mu(\xi)}{4}\right],
\end{align}
where $\nu$ is the complex-valued function $\nu(x) = \left(9/4 - 12(x + \xi) \right)^{1/2}$. There are three relevant cases of interest: (i) For $3/16 - \xi < x$, $\nu(x)$ is purely imaginary, (ii) for $3/16 -\xi = x$, $\nu(x) = 0$ and (iii) for $3/16 -\xi > x$, $\nu(x)$ is real.

\subsubsection{Case (i) $x \in \mathbb{R}^+ \cap (3/16-\xi, \infty) $}

In this case, we choose the square root branch such that $\nu(x) = \ii \left(12(x + \xi) - 9/4\right)^{1/2}$. Set $y = 12(x + \xi) - 9/4$, then we need to find the roots of $g_\xi(y) = g_1(y,\xi)+g_2(y,\xi)$ for $y > 0$ and $\xi < y/12 + 3/16$, where
\begin{subequations}
\begin{align}
g_1(y,\xi) &= \psi\left(\ii y^{1/2}+\frac{3}{2}\right)+\psi\left(\frac{3}{2}-\ii y^{1/2}\right) -\ln \left(-12 \xi +y+\frac{9}{4}\right) \label{g1}\\
g_2(y,\xi) &= \frac{48}{4 y+1} \left(+\frac{360 (1-3 \xi ) \xi -29}{45 (-48 \xi +4 y+9)}+\frac{1}{9} -\xi - \frac{\alpha_\mu(\xi)}{4} \right). \label{g2}
\end{align}
\end{subequations}

The course of action is to analyse the parameter space that allows for roots of $g_\xi$ to exist by examining the behaviour of the functions $g_1$ and $g_2$.

Let us begin by analysing $g_1$. At fixed $y$, $g_1$ is an increasing function of $\xi$ in the relevant domain. It follows from lemma \ref{Lem} in appendix \ref{App:A} that the function $g_1$ is bounded as follows
\begin{align}
|g_1(y, \xi)| \leq \frac{3}{y + 9/4} + \left|\ln \left(y+\frac{9}{4}\right) -\ln \left(-12 \xi +y+\frac{9}{4}\right)\right|.
\label{Ineq}
\end{align}

It follows from bound \eqref{Ineq}\footnote{In fact, it seems to us from numerical analysis that it can be shown that $\psi(z) + \psi(\overline{z}) < \ln (z \overline{z})$, but we haven't been able to provide an analytic argument.} and by studying the definition of $g_1$, eq. \eqref{g1}, that for $\xi \neq 0$ at large values of $y$, the sign of $g_1$ is ${\rm sgn}(g_1) = {\rm sgn}(\ln(y +9/4) - \ln(-12 \xi + y + 9/4))$. This implies that as $y\to \infty$, $g_1(y, \xi) = O(1/y)$ and the limit is approached logarithmically fast for $\xi \neq 0$ and polynomially fast for $\xi = 0$. 

For $g_2$, defined by eq. \eqref{g2}, the sign of the function depends on the ambiguous function $\alpha_\mu$. First, notice that $g_2$ has a pole at 
\begin{equation}
y_{\rm r}= \frac{2160 \alpha_\mu  \xi -405 \alpha_\mu +4320 \xi ^2-1140 \xi +64}{20 (9 \alpha_\mu +36 \xi -4)},
\end{equation}
and hence we exclude this point from the ensuing analysis. Also, at fixed $\xi$, $|g_2|$ is a strictly decreasing function of $y$ bounded as $|g_2(y,\xi)| \leq A/y^2 + B/y$, with $A,B \in \mathbb{R}^+$, and hence as $y \to \infty$, $g_2 = O(1/y)$, and the limit is approached polynomially fast. Further, at large $y$, $y \gg 1$, it holds that sgn$(g_2) = $sgn$(1/9 - \xi - \alpha_\mu(\xi)/4)$, and for $|\alpha_\mu| < |-4 \xi + 4/9|$, sgn$(g_2) = $sgn$(1/9 - \xi)$. In particular, in this case, $g_2 > 0$ if $ \xi < 1/9$ and $g_2 < 0$ if $\xi > 1/9$ for $y \gg 1$.

It follows from the bounds \eqref{Ineq} and the behaviour of the function $g_2$, eq. \eqref{g2}, that at fixed $\xi$ the curves defined by $y \mapsto g_1(y,\xi)$ and $y \mapsto -g_2(y,\xi)$ may intersect in the relevant domain, and hence $g_\xi$ has roots, which are in turn identified with semi-classical solutions in de Sitter spacetime.

From this point, a numerical analysis can be carried out to find solutions. In particular, for $|\alpha_\mu(\xi))| < 4/9$, the minimally coupled field, $\xi = 0$ falls in the case $g_1(y,\xi)<0$ and $g_2(y,\xi) > 0$. We illustrate this particular case with $\alpha_\mu = 0$ for concreteness. In this case, one finds that $y_{\rm root} \approx 1.46912\times 10^{13}$ with an error term $\sim 10^{14}$, with $g(y_{\rm root}) \approx -9.23706 \times 10^{-14}$. This yields in turn that
\begin{equation}
m^2 \simeq 4.89707 \times 10^{12} \Lambda, \text{ for } \xi = 0, \alpha_\mu(\xi) = 0.
\end{equation}

We wish to emphasise that different values for $\alpha_\mu$ and $\xi$ will generally yield roots that fix the cosmological constant at different values.

\subsubsection{Case (ii) $x = 3/16-\xi > 0$}

In case (ii), the curvature coupling is restricted to $\xi_c = 3/16 - x$ and we have to find the roots of
\begin{align}
f_{\xi_c}(x) & = +24 x+\frac{17}{1440 x}-\log (12 x)-\frac{8}{3} +2 \psi \left(\frac{3}{2}\right) -12 \alpha_\mu(3/16-x)
\end{align}
for $x >0$. The existence of roots for $f_{\xi_c}$ strongly depends on the ambiguous function, $\alpha_\mu$. For concreteness, we wish to explore constant-$\alpha_\mu$ solutions. Hence, we set $\alpha_\mu \in \mathbb{R}$.

First, notice that $\lim_{x \to 0^+} f_{\xi_c}(x) \to + \infty$ and $\lim_{x \to \infty} f_{\xi_c}(x) \to + \infty$. Second, analysing the first and second derivatives of $f_{\xi_c}$, one can verify that a minimum is located at $x = x_{\rm min} = 1/720 (15 + 4 \sqrt(30))$, and we have that $f(x_{\rm min}) = -12 \alpha_\mu +4 \sqrt{\frac{2}{15}}+\frac{4}{3}-2 \gamma +\log \left(\frac{1}{68} \left(4 \sqrt{30}-15\right)\right)$, where $\gamma$ is Euler's gamma. 

Hence, there exist solutions to the semi-classical equations if $12  \alpha_\mu \geq  4 \sqrt{\frac{2}{15}}+\frac{4}{3}-2 \gamma +\log \left(\frac{1}{68} \left(4 \sqrt{30}-15\right)\right)$. In the case of equality, there is exactly one solution at with $m^2 = 4 x_{\rm min} \Lambda$. Otherwise, for each value of $\alpha_\mu$ such that $12  \alpha_\mu >  4 \sqrt{\frac{2}{15}}+\frac{4}{3}-2 \gamma +\log \left(\frac{1}{68} \left(4 \sqrt{30}-15\right)\right)$, there are exactly two solutions, one with $m^2 < 4 x_{\rm min} \Lambda$ and one with $m^2 > 4 x_{\rm min} \Lambda$.

\subsubsection{Case (iii) $0< x < 3/16 - \xi$}

Set $z = -12 (x + \xi) + 9/4$ and define for $z \in \mathbb{R}^+ \cup (-12 \xi, -12 \xi + 9/4)$
\begin{align}
h(z,\xi)& = \psi\left(\frac{3}{2}-\sqrt{z}\right)+\psi\left(\frac{3}{2} + \sqrt{z}\right) -\log  \left( \frac{9}{4}-z-12 \xi \right) \nonumber \\
& + \left(\frac{1}{12} \left(\frac{9}{4}-z\right)-\frac{1}{6}\right)^{-1} \left(\frac{1-1080 \left(\xi -\frac{1}{6}\right)^2}{2160 \left(\frac{1}{12} \left(\frac{9}{4}-z\right)-\xi \right)} +\frac{1}{9} -\xi  -\frac{\alpha_\mu }{4}\right)
\end{align}

It follows from the roots and poles of the digamma function, which is holomorphic on $\mathbb{C} \backslash -\mathbb{N}_0$, and for which countably many roots exits on the negative axis between the poles, that for sufficiently large, negative $\xi$ there are several roots for $h$ that define solutions to the semi-classical problem, which can be explored numerically, and in turn fix the admissible values of $\Lambda$ in terms of $m^2$ and $\xi$. 

\section{Final remarks}
\label{sec:Remarks}

We have proposed that, when taking the semi-classical Einstein field equations coupled to quantum matter seriously, the value of the cosmological constant will be determined by the field equations in terms of the parameters of the theory -- the mass and curvature coupling in the case of the Klein-Gordon field -- up to well-known ambiguities. In this sense, the ``very-small" observed cosmological constant of the universe should be constrained by the semi-classical field equations sourced by the stress-energy tensor of the standard model of particle physics, pressumably on a FLRW background to a good approximation, similarly to what occurs in the models that we have studied.

In our Klein-Gordon model, the common folklore that the ``bare" value of the cosmological constant must have a very large contribution from quantum fields that is then cancelled by a fine-tuned counterterm to yield a small ``renormalised" cosmological constant plays no r\^ole in the calculations. Indeed, if one wishes to interpret our results in terms of bare and renormalised quantities, one could interpret semi-classical gravity as providing the renormalised $\Lambda = \Lambda_{\rm ren}$ directly. Much like the bare (formally diverging) stress-energy tensor, the bare cosmological constant plays no r\^ole in the semi-classical gravity equations. Hence, from this viewpoint, the {\it old} cosmological constant problem in its standard formulation \cite{Weinberg:1988cp} is not present, cf. the second paragraph in the Introduction. Further, this letter provides counter-evidence that one could estimate the value of $\Lambda$ to be large. As we exemplified in Sec. \ref{sec:Solutions}, for massless fields, the cosmological constant can take any positive value, while for a massive, minimally-coupled field we have the ratio $m^2/\Lambda \sim 10^{12}$.

In reaching these conclusions, Wald's renormalisation axioms for the stress-energy tensor play a crucial r\^ole. Indeed, if the fourth axiom is not imposed, the right-hand side of eq. \eqref{MinkoAmbLimit} need not vanish, and one is left with an expression that is quite familiar in the literature of the cosmological constant problem, cf. \cite[Eq. (89) and (96)]{Martin:2012bt} up to the ambiguous term proportional to $\alpha_4$. This leads to an important change in the analysis of solutions. Consider the massive case (i). In this case, $g_2$ will not asymptote to zero as $y \to \infty$, but to a constant, changing importantly our results. In any case, in view of \cite[Eq. (89) and (96)]{Martin:2012bt} and of our results, it is clear that the cosmological constant is {\it pure ambiguity} in the framework of \cite{Martin:2012bt} and that the semi-classical equations can partially fix this ambiguity.

The above observations connect the cosmological constant problem, in the sense of the vacuum energy yielding a large contribution, with the violation of Wald's stress-energy renormalisation axioms. 

We have not mentioned anything so far about the {\it new} cosmological constant problem (see e.g. \cite[Sec. 2.3]{Padilla:2015aaa}), as this is an interacting theory problem. For addressing such a matter, perhaps modern techniques of perturbation theory in curved spacetimes should be useful \cite{Rejzner}.

\section*{Acknowledgments}
The author thanks Claudio Dappiaggi for a useful email exchange, in which ref. \cite{Dappiaggi:2008mm} was pointed to the author, and Daniel Sudarsky and Jorma Louko for stimulating conversations. The author is thankful towards Igor Khavkine for a fruitful email exchange concerning renormalisation ambiguities. This work is supported by a DGAPA-UNAM Postdoctoral Fellowship.

\appendix


\section{A lemma for the digamma function}
\label{App:A}

\begin{lemma}
\label{Lem}
Let $z \in \mathbb{C}$, such that $a = \Re z > 0$ and let $b = \Im z$. The following bound holds for the digamma function, $\psi:\mathbb{C} \to \mathbb{C}$,
\begin{align}
|\psi(z) + \psi(\overline{z}) - \ln (z \overline{z})| \leq \frac{2 a}{a^2 + b^2}   
\end{align}
\end{lemma}
\begin{proof}
Using the representation $\psi(z) = \ln z + \int_0^\infty \! dt \, \left(\frac{1}{t} - \frac{1}{1-\ee^{-t}} \right) \ee^{-t z} dt$, we have that $\psi(z) + \psi(\overline{z}) = \ln(z \overline{z}) + \int_0^\infty \! dt \, \left(\frac{1}{t} - \frac{1}{1-\ee^{-t}} \right) \left( \ee^{-t z} + \ee^{-t \overline{z}} \right) dt$. Rearranging $\ee^{-t z} + \ee^{-t \overline{z}} = 2 \ee^{-a t} \cos(b t)$ and using the inequalities $-1 \leq 1/t - 1/\left(1-\ee^t\right) \leq -1/2$ for $t \in \mathbb{R}^+$, one can write $\left| \int_0^\infty \! dt \, \left(\frac{1}{t} - \frac{1}{1-\ee^{-t}} \right) \left( \ee^{t z} + \ee^{t \overline{z}} \right) dt \right| \leq 2 \left| \int_0^\infty \! dt \, \ee^{-a t} \cos(b t) dt \right|$. The integral on the right-hand side can be evaluated directly, yielding
\begin{align}
 \left| \int_0^\infty \! dt \, \left(\frac{1}{t} - \frac{1}{1-\ee^{-t}} \right) \left( \ee^{t z} + \ee^{t \overline{z}} \right) dt \right| \leq \frac{2 a}{a^2 + b^2},
\end{align}
which yields the desired result.

\end{proof}

\end{document}